\documentclass[a4paper,fleqn]{cas-sc}

\usepackage[authoryear,longnamesfirst]{natbib}

\def\tsc#1{\csdef{#1}{\textsc{\lowercase{#1}}\xspace}}
\tsc{WGM}
\tsc{QE}
\tsc{EP}
\tsc{PMS}
\tsc{BEC}
\tsc{DE}

\newtheorem{theorem}{Theorem}
\newtheorem{corollary}[theorem]{Corollary}
\newproof{proof}{Proof}

\usepackage{relsize} 
\usepackage{listings} 
\usepackage[section]{placeins} 

\newcommand{\varhyphen}[1]{{\operatorname{\mathit{#1}}}}

\definecolor{water}{HTML}{CCFFFF}
\definecolor{waterspout}{HTML}{99FFFF}
\definecolor{canary}{HTML}{FFFF99}
\definecolor{menthol}{HTML}{CCFF99}
\definecolor{light-red}{HTML}{FFCCCC}
\definecolor{peach-orange}{HTML}{FFCC99}

\begin{document}
\let\WriteBookmarks\relax
\def\floatpagepagefraction{1}
\def\textpagefraction{.001}

\shorttitle{A statistical significance testing approach for measuring term burstiness}

\shortauthors{S. Sarria Hurdato et~al.}

\title [mode = title]{A statistical significance testing approach for measuring term burstiness with applications to domain-specific terminology extraction}                      
%
%

%
\author[1]{Samuel {Sarria Hurtado}}[orcid=0009-0000-8719-7536]
\credit{Writing – original draft and review \& editing, Methodology, Numerical experiments, Results interpretation, Computer code, Visualization, Proofs}

\affiliation[1]{organization={University of Prince Edward Island},
    addressline={550 University Ave}, 
    city={Charlottetown},
    postcode={C1A 4P3}, 
    state={PE},
    country={Canada}}

\author[1]{Todd Mullen}[orcid=0000-0003-2818-2734]
\credit{Writing – review \& editing, Proofs}

\author[2]{Taku Onodera}[]
\credit{Writing – review \& editing, Notation}

\affiliation[2]{organization={The University of Tokyo},
    addressline={7-chome-3-1}, 
    city={Bunkyo City},
    postcode={113-8654}, 
    state={Tokyo},
    country={Japan}}

\author[1]{Paul Sheridan}[orcid=0000-0002-5484-1951]
\cormark[1]
\ead{paul.sheridan.stats@gmail.com}
\ead[URL]{https://paul-sheridan.github.io/}
\cortext[cor1]{Corresponding author}
\credit{Writing - original draft and review \& editing, Conceptualization, Methodology, Numerical experiments, Results interpretation, Proofs, Supervision, Funding acquisition}

\begin{abstract}
A term in a corpus is said to be ``bursty'' (or overdispersed) when its occurrences are concentrated in few out of many documents. In this paper, we propose Residual Inverse Collection Frequency (RICF), a statistical significance test inspired heuristic for quantifying term burstiness. The chi-squared test is, to our knowledge, the sole test of statistical significance among existing term burstiness measures. Chi-squared test term burstiness scores are computed from the collection frequency statistic (i.e., the proportion that a specified term constitutes in relation to all terms within a corpus). However, the document frequency of a term (i.e., the proportion of documents within a corpus in which a specific term occurs) is exploited by certain other widely used term burstiness measures. RICF addresses this shortcoming of the chi-squared test by virtue of its term burstiness scores systematically incorporating both the collection frequency and document frequency statistics. We evaluate the RICF measure on a domain-specific technical terminology extraction task using the GENIA Term corpus benchmark, which comprises 2,000 annotated biomedical article abstracts. RICF generally outperformed the chi-squared test in terms of precision at k score  with percent improvements of 0.00\% (P@10), 6.38\% (P@50), 6.38\% (P@100), 2.27\% (P@500), 2.61\% (P@1000), and 1.90\% (P@5000). Furthermore, RICF performance was competitive with the performances of other well-established measures of term burstiness. Based on these findings, we consider our contributions in this paper as a promising starting point for future exploration in leveraging statistical significance testing in text analysis.
\end{abstract}



\begin{keywords}
burstiness \sep inverse collection frequency \sep inverse document frequency \sep terminology extraction \sep significance testing
\end{keywords}

\maketitle

\section{Introduction}\label{sec:intro}
Not all terms occurring an equal number of times in a collection of textual documents are equally good for characterizing content. In a classical illustration of this point, \cite{Church1995} consider a collection of Associated Press Newswire articles in which the terms \emph{boycott} and \emph{somewhat} each occur roughly 1,000 times. They nevertheless point out that \emph{boycott} serves as a more effective keyword than \emph{somewhat}. The reason is that \emph{boycott} is a topical term the occurrences of which are statistically concentrated in a modestly sized subset of articles about consumer activism, whereas \emph{somewhat} is a low information term the occurrences of which are spread out across all of the articles in the collection. The term \emph{boycott} exhibits a relatively high degree of a quality known as ``burstiness''.

Term burstiness quantifies ``how peaked a word’s usage is over a particular corpus of documents''~\cite[p. 11]{Santini2021}. In a recent review, \cite{Sonning2022} surveyed common empirical measures of term burstiness, including key measures advanced by~\cite{Juilland1964}, \cite{Church1995,Church1999}, \cite{Katz1996}, \cite{Kwok1996}, \cite{Schafer2002}, \cite{Gries2008}, and~\cite{Irvine2017}. Term burstiness is a statistical feature of some importance in corpus profiling~(\cite{Schafer2002,Pierrehumbert2012,Irvine2017,Sharoff2017,Santini2018,Santini2021}). In the field of textual information retrieval, term burstiness has be leveraged for keyword extraction~(\cite{Church1995,Katz1996,Church1999,Matsuo2004}, document classification~(\cite{Madsen2005}), document retrieval~(\cite{Cummins2015,Cummins2017,Bahrani2023,Podder2023}), and sentiment analysis~(\cite{Najar2023}). While filtering out bursty terms when compiling term indexes is commonplace~(\cite{Sharoff2017}), they show promise as identifying well with technical terminology~(\cite{Santini2021}), a main application of which being the collection-level extraction of technical terms for building domain terminologies~(\cite{Firoozeh2020}). In image retrieval, \cite{Jegou2009} and \cite{Kato2022} have used the concept of burstiness to analyze visual words with encouraging results. Burstiness has also been employed for distinguishing human-authored text from text produced by large language models~(\cite{Wang2023}).

In this paper, we investigate the potential for statistical significance testing to improve characterizing term burstiness in the context of a $d>0$ document collection based on an $m>0$ term vocabulary. Term burstiness measures typically exploit one or both of two foundational word statistics. The one, the \emph{total document frequency} (TDF), denoted by~$b_i$, is the number of documents in a collection of interest containing the $i$'th term ($i=1,\ldots, m$). It is the basis for the famed \emph{inverse document frequency} (IDF) statistic $\log(d/b_i)$, where~$d$ is the number of documents in the collection. The other, the \emph{total term frequency} (TTF), denoted by~$n_i$, is the number of occurrences of the $i$'th term in the collection. \cite{Kwok1990} proposed the associated \emph{inverse collection frequency} (ICF) word statistic $\log(n/n_i)$. The value $n$ denotes the total number of terms in the collection. The heuristic proposed by~\cite{Church1995}, for instance, quantifies term burstiness as the ratio $n_i/b_i$. The chi-squared significance test, another common approach for quantifying term burstiness, incorporates TTF (or equivalently, ICF) but not TDF (or equivalently, IDF). It is, however, well-understood that TDF (or equivalently, IDF) is informative for assessing term burstiness.

We advance the first ever significance test for term burstiness exploiting the TDF word statistic (or equivalently, IDF). The test, which also depends on the TTF (or equivalently, ICF) word statistic, assigns term burstiness scores according to the negative logarithm of P-values representing deviations from what would be expected under the multinomial language model. While the test constitutes a theoretical advance, the calculation of test P-values is at present computationally infeasible outside of toy example sized corpora. This reality leads us to propose a test inspired heuristic, which we call Residual ICF (RICF), that approximates the test P-value for the $i$'th term as the null distribution expected value minus the observed ICF value. By leveraging both ICF and IDF, we hypothesize that the RICF measure offers a practical advance in term burstiness quantification in a significance testing framework.

To assess the effectiveness of the RICF measure, we evaluated it against established measures on a molecular biology terminology extraction task using the GENIA Term corpus benchmark of~\cite{Kim2003}. We also included the state-of-the-art keyword extractor KeyBERT, as recently advanced by~\cite{Grootendorst2020}, in our evaluations. Our findings show RICF to be competitive with established methods with no single method yielding the best performance across all experimental settings. In an additional analysis, we show that RICF outperforms established term burstiness measures at filtering out stopwords. These encouraging findings hint that future efforts to better approximate test P-values are liable to result in even more effective measures of term burstiness.

In formulating the RICF measure, we derived what to our knowledge is a hitherto unreported formula connecting ICF and IDF in expectation under the multinomial language model. The formula is required for evaluating RICF scores, and may prove to be of independent interest to text analysis practitioners.

The rest of the paper is organized as follows. Necessary background is provided in Section~\ref{sec:related-work}. Section~\ref{sec:research-objectives} states concrete research objectives. Section~\ref{sec:genia} introduces the GENIA Term corpus benchmark dataset. The exact significance test for term burstiness, and related RICF measure, is presented in Section~\ref{sec:methodology}. The formula relating IDF to ICF in expectation is also given in this section. Section~\ref{sec:results} presents the experimental design and results of the terminology extraction task on the GENIA Term corpus dataset. Section~\ref{sec:discussion} discusses practical and theoretical implications of our contributions together with ideas for future work. A summary of the paper is provided in Section~\ref{sec:conclusion}.

\section{Research objectives}\label{sec:research-objectives}
The chi-squared test, when deployed in a language modelling setting as done by~\cite{Matsuo2004}, is a traditional approach to quantifying term burstiness. The chi-squared test burstiness score for the $i$'th term in a collection depends on the TTF value (i.e., the $n_i$ term) of the $i$'th term, but not its TDF value (i.e., the $b_i$ term). TDF is, however, an information rich quantity that has been successfully exploited in traditional term burstiness heuristics~(\cite{Church1995}, and \cite{Irvine2017}). At the same time, the significance testing paradigm offers promise when it comes to the development of new and effective term burstiness measures, as significance tests use accurate and readily interpretable P-values (i.e., tail probabilities) to quantify term burstiness, as opposed to heurististically motivated formulas.

The main aim of this research is to develop an effective alternative to the traditional chi-squared test for quantifying term burstiness, an alternative rooted in statistical significance testing. Specifically, our research aims are outlined as follows:

\noindent\textbf{Research Objective 1} (RO1): Develop the first-ever significance test-based measure of term burstiness leveraging both TTF/ICF and TDF/IDF word statistics, and investigate the feasibility of using it in real-world text analysis applications.

\noindent\textbf{Research Objective 2} (RO2): Relate the foundational text analysis quantities IDF and ICF in expectation under the multinomial language model.

\noindent\textbf{Research Objective 3} (RO3): Using the relation from R02, develop a computationally efficient heuristic for test P-values that is practical for real-world text analysis applications.

\noindent\textbf{Research Objective 4} (RO4): Evaluate the performance of our proposed significance test inspired term burstiness heuristic on a technical terminology extraction task.

\noindent\textbf{Research Objective 5} (RO5): Evaluate the performance of our proposed significance test inspired term burstiness heuristic its ability to identify stopwords.

The primary novelty of this work lies in the application of sophisticated significance tests in the field of text analysis. Our contributions should prove of interest to the research community working on term weighting schemes, and information science practitioners in general. It is our hope that this work will serve as a testbed for the future exploration of leveraging significance tests in the field of text analysis.


\section{Related work}\label{sec:related-work}

\subsection{Modeling framework}\label{subsec:model}
In this section, we settle the notation used throughout the paper for textual document representation and modeling. Table~\ref{tab:notation} sets out the bulk of the notation used in this work.

\begin{table*}[th]
\centering
\begin{tabular}{p{2cm}p{12cm}}
Symbol & Definition \\
\midrule
$\mathcal{C}_{d,m}$ & A \emph{collection}, or \emph{corpus}, of $d$ documents made up of $m$ distinct terms \\ 
$d$ & \emph{Collection size}: Number of documents in the collection \\
$m$ & \emph{Vocabulary size}: Number of distinct terms from which the collection documents are composed \\
$n_{ij}$ & \emph{Term frequency} (TF): Number of times the $i$'th term occurs in the $j$'th document \\
$[n_{ij}]_{m \times d}$ & \emph{Term-document matrix}: An $m \times d$ matrix of TFs \\
$n_i = \sum_{j=1}^d n_{ij}$ & \emph{Total term frequency} (TTF): Number of times the $i$'th term occurs in the collection \\
$n_j = \sum_{i=1}^m n_{ij}$ & \emph{Document size}: Number terms in the $j$'th document, including multiplicities \\
$n = \sum_{j=1}^d n_j$ & Number of terms in the collection, including multiplicities  \\
$b_{ij}$ & \emph{Binary term frequency} (BTF): Indicator variable that is $1$ if the $i$'th term occurs in the $j$'th document and $0$ otherwise \\
$b_{i} = \sum_{j=1}^d b_{ij}$ & \emph{Total document frequency} (TDF): Number of documents containing the $i$'th term \\
$df(i) = b_i / d$ & \emph{Document frequency} (DF): Proportion of documents containing the $i$'th term \\
$cf(i) = n_i / n$ & \emph{Collection frequency} (CF): Proportion of the $i$'th term in the collection \\
$\mu$ & Expected document length under the multinomial language model \\
$\theta_{ij}=\theta_i (\forall j)$ & Probability, under the multinomial language model, of the $i$'th term in $j$'th document subject to the constraint $\sum_{i=1}^m \theta_i = 1$ \\
$\rho_i$ & Shorthand for $1 - \theta_i$ \\
$\phi_{ij}$ & The probability $1 - (1 - \theta_i)^{n_j}$ that, under the multinomial language model, the $i$'th term occurs at least once in the $j$'th document  \\ 
$\psi_i$ & Shorthand for $e^{\mu \theta_i} - 1$\\
\midrule
\end{tabular}
\caption{Multinomial language model for document modeling definitions and notation.}
\label{tab:notation}
\end{table*}

Consider a collection, $\mathcal{C}_{d,m}$, of $d>0$ documents composed from a vocabulary of $m>0$ distinct terms. We assume the multinomial language model (a.k.a., bag-of-words model or unigram model) for document representation. In this model, the $j$'th document ($1 \leq j \leq d$) is represented as the vector of counts $\langle n_{1j}, \ldots, n_{mj} \rangle$, where $n_{ij} \geq 0$ denotes the number of times the $i$'th term ($1 \leq i \leq m$) occurs in the $j$'th document. The quantity $n_{ij}$ is called the \emph{term-in-document frequency}, or \emph{term frequency} (TF) for short, of the $i$'th term in the $j$'th document. TFs are organized in an $m$~by~$d$ \emph{term-document matrix}, $[n_{ij}]_{m \times d}$, each row of which corresponds to a term, and each column to a document. This matrix representation of a collection of documents is a textbook starting point for analysis.

The multinomial language model takes the document size $n_j$ ($1 \leq j \leq d$) as a fixed number of trials. Assume $N_{ij}$ ($1 \leq i \leq m$) to be a random variable denoting the number of occurrences of the $i$'th term in the $j$'th document. According to the model, the vector of TFs $\langle N_{1j}, \ldots, N_{mj} \rangle$ for the $j$'th document is multinomially distributed according to
\begin{equation*}
Pr(N_{1j}=n_{1j}, \ldots, N_{mj}=n_{mj}) = \frac{n_j!}{\prod_{i=1}^m n_{ij}!} \prod_{i=1}^m \theta_{i}^{n_{ij}},
\end{equation*}
where $n_j=\sum_{i=1}^m n_{ij}$ is the size of the $j$'th document, and $\langle \theta_1, \ldots, \theta_m \rangle$ is a vector of Bernoulli success probabilities satisfying $\sum_{i=1}^m \theta_i = 1$.

\subsection{Term weighting models}\label{subsec:term-weightings}
A \emph{term weighting model}, or \emph{term weighting}, is a function assigning a nonnegative, real-valued score to the $i$'th term in the $j$'th document in a collection. The higher a term scores in a document, the better it is considered to represent the document’s content. The development of novel term weighting models is an active area of research in text analysis~\cite{Rathi2023}.

\emph{Term frequency–inverse document frequency} (TF–IDF) is a foundational term weighting model. Proposed by~\cite{Salton1973}, it scores the importance of the $i$'th term in the $j$'th document according to the formula $\varhyphen{tf-idf}(i,j) = n_{ij} \times \log(d/b_i)$. In so doing, the TF-IDF score $\varhyphen{tf-idf}(i,j)$ represent the importance of the $i$'th term specifically within the $j$'th document while also considering the $i$'th term's importance across the entire collection. TF–IDF is an archetypal example of the general term weighting functional form: the product of a local weight (e.g., $n_{ij}$), a global weight (e.g., $\log(d/b_i)$), and a document length normalization factor (e.g., trivially $1$ in this case); see~\cite{Polettini2004} for more details.

\subsubsection{Local term weighting models}
The local weighting term quantifies the importance of the $i$'th term to the $j$'th document without regard to other documents in the collection. TF and its binarization, BTF, are two classical local weightings~( see~\cite{Rathi2023}). Common local weightings were recently surveyed by~\cite{Domeniconi2016} and~\cite{Ghahramani2021}. \cite{Dogan2019} argue that the choice of local weighting model plays an important role in term weighting function design.

\subsubsection{Global term weighting models}
Unlike local term weightings, which operate on a document-level scope, global term weightings take into account the entire document collection to quantify term importance~(\cite{Alshehri2023}). Two traditional global term weightings factor prominently into this work: the IDF and ICF word statistics. 

\textbf{IDF}: The \emph{inverse document frequency} (IDF) of the $i$'th term in a collection is defined as
\begin{equation*}
idf(i) = \log(d / b_i).
\end{equation*}
Never mind that the IDF formula specifies a logarithmically scaled inverse
proportion, and not an inverse frequency. IDF has been a workhorse in text analysis ever since it was advanced by~\cite{SparckJones1972} some half-century ago. The intuition behind IDF is that terms occurring in few documents tend to be more informative and can help distinguish documents from each other. Terms that are common across many documents (e.g., \emph{the}, \emph{of}, \emph{are}, etc.) provide comparatively lower discriminative power and are assigned lower IDF~scores.

\textbf{ICF}: Proposed by~\cite{Kwok1990}, the \emph{inverse collection frequency} (ICF) of the $i$'th term in a collection is defined as
\begin{equation*}
icf(i) = \log(n / n_i),
\end{equation*}
where $n/n_i$ is reciprocal of the proportion the $i$'th term makes up out of all the terms in the collection, counting multiplicities. This IDF cousin has been successfully deployed in language modeling~(\cite{Ponte1998}), document clustering~(\cite{Reed2006}), and document retrieval~(\cite{Amati2002,Kwok1996,Pirkola2001,Pirkola2002,Jimenez2018}).

Term weighting models are generally devised by pairing an existing local weighting together with a novel global one, typically appealing to cosine similarity as a normalizing factor~(\cite{Rathi2023}). Underlying term weighting model construction is a plug-and-play principle: namely, substitute the IDF term in TF–IDF with an alternate global weighting, and evaluate the new term weighting model on such standard text analysis tasks as document retrieval or document classification. TF–IDF variants have been advanced in supervised~(\cite{Martineau2009,Ren2013,Chen2016,Domeniconi2016,Dogan2019,Carvalho2020,Alshehri2023}) and unsupervised~(\cite{Robertson2009,Sabbah2017,Jimenez2018,Alsmadi2019,Campos2020}) contexts.

\subsection{Term burstiness measures} \label{subsec:bursty-measures}
A burstiness measure amounts to a global term weighting that quantifies term frequency dispersion at the collection level. IDF has been marshalled as a simple measure of term burstiness~(\cite{Schafer2002,Irvine2017}), as has ICF~(\cite{Kwok1996}). More elaborate measures have been proposed by~\cite{Juilland1964}, \cite{Church1995,Church1999}, \cite{Katz1996}, \cite{Gries2008}, and~\cite{Irvine2017}. In what follows, we describe the measures of term burstiness that are used in this study. These measures we summarize in Table~\ref{tab:bursty-measures} for convenience. 

\begin{table*}[ht]
\setlength{\tabcolsep}{4pt}
\centering
\begin{tabular}{lllr}
Name & Abbreviation & Formula & References \\
\midrule
Chi-squared test & Chi-sq & $\chi_{d-1}^2(i) = \sum_{j=1}^d \frac{(n_{ij} - n_i / d)^2}{n_i / d}$ & \cite{Matsuo2004} \\
Church and Gale & CG & $cg(i) = n_i / b_i$ & \cite{Church1995} \\
Irvine and Callison-Burch & ICB & $icb(i) = \frac{1}{b_i}\sum_{j=1}^d \frac{n_{ij}}{n_j}$ & \cite{Irvine2017} \\
Deviation of Proportions & DoP & $dop(i) = 1 - \frac{1}{2}\sum_{j=1}^d \lvert \frac{n_{ij}}{n_i} - \frac{n_j}{n} \rvert$ & \cite{Gries2008} \\
\midrule
\end{tabular}
\caption{Published measures of term burstiness used in this study.}
\label{tab:bursty-measures}
\end{table*}

\textbf{Chi-squared test} (Chi-sq): \cite{Gries2008} details a classical adaptation, tracing back to~\cite{Nagao1976} and popularized by~\cite{Matsuo2004}, of the chi-squared test of statistical significance to term burstiness quantification. Under the null hypothesis, the $i$'th term is assumed to be uniformly distributed across the $d$ documents in the corpus, making $n_i/d$ the expected number of occurrences of the $i$'th term in $j$'th document. For large samples, the P-value associated with the chi-squared test statistic
\begin{equation*} \label{eqn:chisq-score}
\chi_{d-1}^2(i) = \sum_{j=1}^d \frac{(n_{ij} - n_i / d)^2}{n_i / d}
\end{equation*}
with $d-1$ degrees of freedom. The test quantifies the degree to which the observed TFs for the $i$'th term (i.e., $n_{i1}, n_{i2}, \ldots, n_{id}$) deviate from uniformity. The larger the negative logarithm of the P-value, the higher the bursty score.

\textbf{Church and Gale} (CG): In early work, \cite{Church1995} empirically quantify the burstiness of the $i$'th term in a collection as its TTF to BTF ratio:
\begin{equation*} \label{eqn:cg-score}
cg(i) = n_i / b_i.
\end{equation*}
The intuition behind the formula being that a bursty term will occur many times (i.e., $n_i$ large) in few documents (i.e., $b_i$ small), resulting in a large score (i.e., $n_i / b_i$). Conversely, frequently occurring words with low information content ought to be assigned low scores. The measure is rank equivalent to the difference of IDF and ICF: $idf(i) - icf(i)=\log(n_i/b_i)-\log(n/d)$.

\textbf{Irvine and Callison-Burch} (ICB): \cite{Irvine2017} proposed this modification of the CG measure
\begin{equation*} \label{eqn:icb-score}
icb(i) = \frac{1}{b_i}\sum_{j=1}^d \frac{n_{ij}}{n_j}
\end{equation*}
to account for varying document lengths.

\textbf{Deviation of Proportions} (DoP): Finally, \cite{Gries2008} advanced the term burstiness measure
\begin{equation*} \label{eqn:gries}
dop(i) = 1 - \frac{1}{2}\sum_{j=1}^d \biggl\lvert \frac{n_{ij}}{n_i} - \frac{n_j}{n} \biggl\rvert
\end{equation*}
as an improvement over the archaic measures of~\cite{Juilland1964}, \cite{Carroll1970}, \cite{Rosengren1971}, and~\cite{Engwall1974}.

We direct the reader to the work of~\cite{Gries2008} and of~\cite{Sharoff2017} for more in depth reviews on the subject word frequency dispersion quantification.

\subsection{Domain-specific terminology extraction}
Terminology extraction methods aim to automatically pinpoint and extract specialized terms in textual corpora composed of documents derived from a particular subject area~(\cite{Wang2016}). Well-extracted terms are technical in nature and hold significant importance within their respective domains~(\cite{Nomoto2022}). Extracted domain-specific terms have been used for such downstream text analysis applications as clinical text processing by \cite{Zhang2016}, sentiment analysis by \cite{Abulaish2022}, and knowledgebase construction by \cite{Khosla2019}.

Domain-specific term extraction is strongly tied to conventional document-specific keyword extraction. For example, the chi-squared test of \cite{Matsuo2004} was originally developed for identifying document keywords by representing a document of interest as a corpus of its sentences. Only later was the test repurposed to identify corpus-level terms~(\cite{Gries2008,Zhang2016}). This is generally possible of keyword extraction methods. \cite{Zhang2016}, for instance, adapted the keyword extractor RAKE of \cite{Rose2010} to act at the collection level by replacing any document-level term counts in the methodology by collection-level ones. We take a similar tack, described below, in applying the keyword extractor KeyBERT of \cite{Grootendorst2020} to corpora. \cite{Abulaish2022} review the current state of domain-specific keyword extraction methods. In their recent review, \cite{Nadim2023} find that KeyBERT outperforms other state-of-the-art keyword extraction methods.

\section{Data} \label{sec:genia}
The GENIA Term corpus, developed by \cite{Kim2003}, is a collection of annotated biomedical texts that is widely used as a benchmark to evaluate keyword extraction tools~(\cite{Tran2023}). In our experiments, we use the latest release, Version 3.02, which consists of 2,000 article abstracts extracted from the MEDLINE database. The corpus comprises roughly 450,000 total tokens, out of which 98,748 are annotated as biological terms.

\begin{table*}[ht]
\centering
\begin{tabular}{lcrrrr}
Semantic class & Class color & Sub-classes & Unique terms & Annotations \\
\midrule
\textsf{Amino Acid} & \cellcolor{water} water & 15 & 10,159 & 42,478 \\
\textsf{Nucleotide} & \cellcolor{waterspout} waterspout & 12 & 5,574 & 11,619 \\
\textsf{Multi-cell} & \cellcolor{canary} canary & 5 & 1,444 & 5,247 \\
\textsf{Cell} & \cellcolor{menthol} menthol & 4 & 4,051 & 11,626 \\
\textsf{Other} & \cellcolor{light-red} light-red & 1 & 10,560 & 19,999 \\
\midrule
Total & & 37 & 31,788 & 90,969
\end{tabular}
\caption{GENIA Term corpus color-coded, high-level semantic classes, numbers of sub-classes, and post-preprocessing biological term and annotation counts.}
\label{tab:genia-stats}
\end{table*}

Table~\ref{tab:genia-stats} shows an overview of the GENIA Term corpus data. Annotations are labeled according to 37 semantic concepts (i.e., sub-classes), which are grouped into six color-coded semantic classes, representing high-level biological constituents. The biological term and annotation counts reported in the table are those obtained after running the preprocessing steps described in Section~\ref{subsec:preprocessing}.

\section{Methodology}\label{sec:methodology}
Here we describe a novel statistical significance testing strategy for quantifying term burstiness. The strategy makes use of an equation, which we derive, relating IDF to ICF in expectation that may prove of some independent interest to the reader.

\subsection{An exact test of statistical significance for term burstiness}\label{subsec:bursty-test}
\cite{Madsen2005} points out that the multinomial language model does not account for term burstiness. It is natural, however, to model term burstiness in a statistical significance testing framework that takes the TFs in a corpus being multinomially distributed as null hypothesis. In what follows we propose such a test.

The basic idea behind the test is to compare the observed value of $n_i$ with its expected value conditional on~$b_i$ under the multinomial language model. The null hypothesis is $\mathcal{H}_0: n_i = \mathbf{E}[N_i=n_i|B_i=b_i]$, where $N_i$ and $B_i$ are random variables denoting the TF and TDF of the $i$'th term, respectively. The P-value associated with observing an outcome at least as extreme as the conditionally expected one specified in $\mathcal{H}_0$ is given by the tail probability
\begin{equation} \label{eqn:bursty-tail-prob}
  Pr(N_i \geq n_i | B_i = b_i)= \frac{\mathlarger{\sum}_{k=n_i}^{n} \mathlarger{\sum}_{u=1}^{U_i(k)} \mathlarger{\prod}_{j=1}^{d} \mathlarger{\sum}_{v=1}^{V_j(u)} \frac{n_j!\theta_1^{y_{1j}^{(v)}} \cdots \theta_{i-1}^{y_{i-1,j}^{(v)}}\theta_i^{x_{ij}^{(u)}}\theta_{i+1}^{y_{i+1,j}^{(v)}}\cdots \theta_m^{y_{mj}^{(v)}}}{y_{1j}^{(v)}!\cdots y_{i-1,j}^{(v)}!x_{ij}^{(u)}!y_{i+1,j}^{(v)}!\cdots y_{mj}^{(v)}!}}{\sum_{A \in F(b_i)} \prod_{r \in A}\phi_{ir} \prod_{s \in A^c}(1-\phi_{is})}.
\end{equation}
This constitutes an exact test of statistical significance. The formula amounts to summing the probability mass of all term-document matrices satisfying the condition $N_i \ge n_i | B_i = b_i$. The calculation of this probability is somewhat involved, as Eq.~(\ref{eqn:bursty-tail-prob}) attests.

The numerator in Eq.~(\ref{eqn:bursty-tail-prob}) is the sum of joint probabilities $\sum_{k=n_i}^n Pr(N_i = k, B_i = b_i)$. The vector $\langle x_{i1}^{(u)},\ldots,x_{ij}^{(u)},\allowbreak\ldots,x_{id}^{(u)} \rangle$ is the $u$'th integer solution ($1 \leq u \leq U_i(k)$) to
\begin{equation*}
k = x_{i1}^{(u)} + \ldots + x_{ij}^{(u)} + \ldots + x_{id}^{(u)},
\end{equation*}
where $k$ is an integer from $n_i$ to $n$, $0 \leq x_{ij} \leq k$ for $1 \leq j \leq d$, and exactly $b_i$ of the $x_{ij}$'s are nonzero. The $U_i(k)$ term is the number of solutions to the equation for a given $i$ and $k$. The related vector $\langle y_{1j}^{(v)},\ldots,y_{i-1,j}^{(v)},y_{i+1,j}^{(v)},\ldots,y_{mj}^{(v)} \rangle$ is the $v$'th integer solution ($1 \leq v \leq V_j(u)$) to
\begin{equation*}
n_j = y_{1j} + \ldots + y_{i-1,j} + x_{ij}^{(u)} + y_{i+1,j} + \ldots + y_{mj}
\end{equation*}
such that $0 \leq y_{\ell j} \leq min(n_\ell,n_j)$ for $\ell = 1, \ldots, i-1, i+1,\ldots,m$. The term $V_j(u)$ is defined similarly to $U_i(k)$.

The denominator in Eq.~(\ref{eqn:bursty-tail-prob}) is the probability $Pr(B_i = b_i)$ that the $i$'th term occurs in exactly $b_i$ documents. The random variable $B_i$ is Poisson binomially distributed (see Appendix~\ref{appendix:eidf}). In the PMF for $B_i$, $F(b_i)$ denotes the set of all subsets of $b_i$ integers that can be selected from $\mathbb{Z}_d = \left\{1,2,\ldots,d\right\}$. The set $A$ is an element of $F(b_i)$, and $A^c$ is the complement of $A$ with respect $\mathbb{Z}_d$. The term $\phi_{ij} = 1 - \rho_i^{n_j}$ is the probability that the $i$'th term occurs at least once in the $j$'th document.

\subsection{A relationship connecting IDF and ICF in expectation} \label{subsec:icf-idf-relation}
Given the combinatorial nature of the test P-value, its evaluation is computationally intractable outside the smallest of toy examples. To address this challenge, we derive an equation relating IDF to ICF in expectation, under the multinomial model, that we leverage in constructing a natural test P-value proxy.

Under the multinomial model, we find that the expected value of IDF is related to that of ICF as
\begin{equation} \label{eqn:idf-icf-expected-relation}
\mathbf{E}[icf(i)] \approx \mathbf{E}[idf(i)] - \log(\theta_i / \psi_{i}) - \frac{\mu\theta_i - (1 - \theta_i)\psi_{i}}{2d\mu \theta_i \psi_{i}} - \mu \theta_i + o(d^{-2}),
\end{equation}
where $\psi_{i} = e^{\mu \theta_i} - 1$ and $\mu=\mathbf{E}[N_j]$ is the expected document length. The relation uses the approximation $(1-\theta_i)^{n_j} \approx e^{-n_j \theta_i} \approx e^{\mu\theta_i}$. This is reasonable when $\theta_i \approx 0$ and the document lengths are well-behaved. Exact formulae for $\mathbf{E}[icf(i)]$ and $\mathbf{E}[idf(i)]$ are derived in Appendix~\ref{appendix:eicf} and Appendix~\ref{appendix:eidf}, respectively. The main trick is to apply the second order Taylor expansion $\mathbf{E}[\log(X)] \approx \log\mathbf{E}[X] - \mathbf{V}[X]/(2\mathbf{E}[X]^2)$ to each of $\mathbf{E}[idf(i)]=\mathbf{E}[\log(d / B_i)]$ and $\mathbf{E}[icf(i)]=\mathbf{E}[\log(n / N_i)]$. Deriving Eq.~(\ref{eqn:idf-icf-expected-relation}) from $\mathbf{E}[icf(i)]$ and $\mathbf{E}[idf(i)]$ is an exercise in basic algebraic manipulations.

\begin{figure}[!ht]
\centering
\includegraphics[width=0.9\columnwidth]{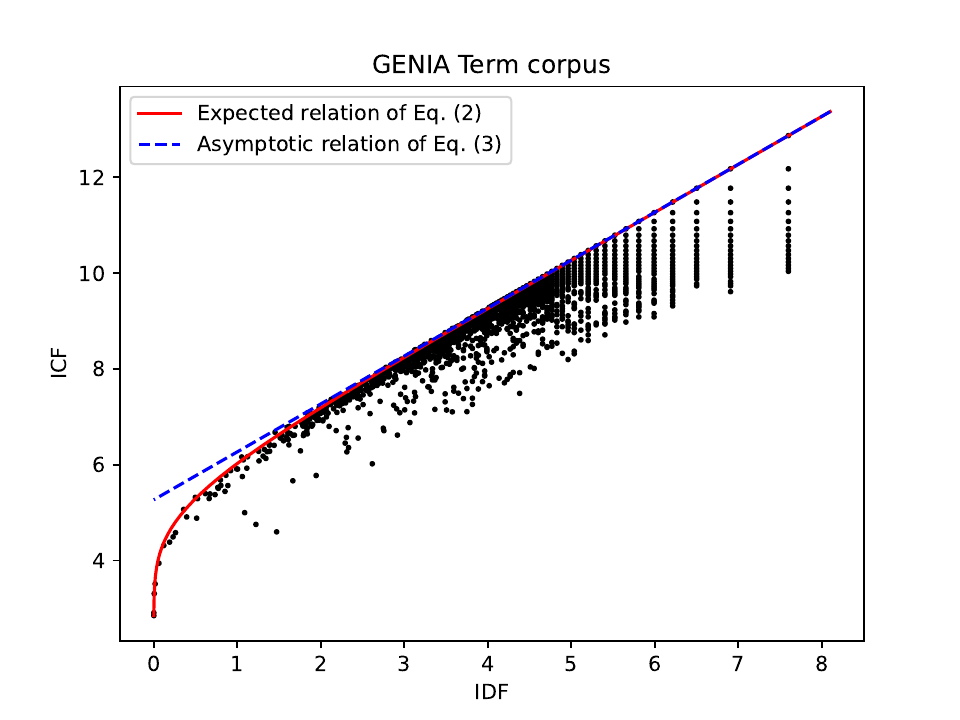}
\caption{GENIA Term corpus IDF versus ICF plot.}
\label{fig:genia-idf-vs-icf}
\end{figure}

It is interesting to note that the relation reduces to the simple linear function
\begin{equation} \label{eqn:idf-icf-expected-relation-asymp}
\mathbf{E}[icf(i)] \approx \mathbf{E}[idf(i)] + \frac{\mu / 2 - 1}{2d\mu} + \log(\mu) + o(d^{-2})
\end{equation}
in the limit as $\theta_i$ approaches $0$ (i.e., for vanishingly rare terms).

In Figure~\ref{fig:genia-idf-vs-icf}, the expected relations of Eqs.~(\ref{eqn:idf-icf-expected-relation})~and~(\ref{eqn:idf-icf-expected-relation-asymp}) are verified on the GENIA Term corpus dataset. Empirical term IDF scores are plotted against their corresponding ICF scores (black dots). The expected curve of Eq.~(\ref{eqn:idf-icf-expected-relation}) is overlain red, and the asymptotic expected curve of Eq.~(\ref{eqn:idf-icf-expected-relation-asymp}) in blue. Outlier terms falling below the red curve are commonly thought of as bursty terms.

\subsection{Significance test inspired term burstiness heuristic}\label{subsec:butsty-term-heuristic}
The negative logarithm of the significance test P-value is a natural choice of term burstiness measure, as it takes on values in $[0,\infty)$ with the more bursty the term, the higher the score. As we have stated, however, the computation of test P-values is generally intractable.

Motivated by practicality, we propose the residual-like heuristic, which we dub Residual ICF (RICF), 
\begin{equation} \label{eqn:ricf}
  ricf(i) = \mathbf{E}[icf(i)|B_i=b_i] - icf(i)
\end{equation}
for quantifying the burstiness of the $i$'th term in a collection. This difference between the conditionally expected ICF score of the $i$'th term, conditional on~$b_i$, and the observed one serves as a workable proxy for the negative logarithm of~Eq.~(\ref{eqn:bursty-tail-prob}).

The primary merit of Eq.~(\ref{eqn:ricf}) is that it is efficient to compute. Appendix~\ref{appendix:cond-expectation-algorithm} describes a simple numerical estimation method for approximating $\mathbf{E}[icf(i)|B_i=b_i]$. That terms scoring more highly are definitionally more bursty is a convenient bonus. A demerit is that by incorporating only information from the first moment of $N_i|B_i=b_i$ into RICF, the generated rankings are bound to be less accurate than those obtained from computing the corresponding P-value negative logarithm.

\section{Experimental results}\label{sec:results}
In this section, we evaluate our proposed term burstiness measure at a variety of domain-specific terminology extraction tasks using the GENIA Term corpus benchmark.

\subsection{Data preprocessing} \label{subsec:preprocessing}
We performed a number of preprocessing steps on the raw GENIA Term corpus data to prepare it for analysis. The preprocessing code is implemented in Python~3.10.12. In addition to standard Python libraries, we made use of the natural language processing libraries \textbf{nltk}~3.8.1, \textbf{bs4}~0.0.1, \textbf{inflect}~7.0.0, \textbf{html5lib}~1.1, and \textbf{lxml}~4.9.3. The total number of terms after preprocessing is $n=362,587$, and the vocabulary size $m=40,804$. In all, $90,969$ terms are labeled from among $31,788$ distinct biological terms. 

The initial data preprocessing step involved extracting article texts from the corpus XML file. Article entries contain a MEDLINE~ID, title, and sentence segmented abstract text. For each article, we merged the title text and abstract text. We removed the sentence ``(ABSTRACT TRUNCATED AT 250 WORDS)'' from 24 abstracts, and ``(ABSTRACT TRUNCATED AT 400 WORDS)'' from three abstracts.

Preprocessing encoded biological terms proved a tricky operation. Consider the following annotated sentence taken from the article with MEDLINE ID 95280917:
\begin{lstlisting}[belowskip=-1.0 \baselineskip]
In this report, we show that the minimal <cons lex="CD4_promoter" sem="G#DNA_domain_or_region"><cons lex="CD4" sem="G#protein_molecule">CD4</cons> promoter</cons> has four <cons lex="factor_binding_site" sem="G#DNA_domain_or_region">factor binding sites</cons>, each of which is required for full function.
\end{lstlisting}
Technical biological terms are marked-up with \textsf{cons} (i.e., biological constituent) tags. The contiguous text \emph{factor binding sites} illustrates the simple case when a sequence of consecutive terms, or even a single term, is enclosed by \textsf{cons} tags. Here, the text \emph{factor binding sites} is marked-up with the \textsf{lex} (i.e., lexical unit) tag factor\_binding\_site, where the associated \textsf{sem} (i.e., semantic class) tag flags factor\_binding\_site as member of the G\#DNA\_domain\_or\_region semantic class. However, many biological terms are recursively annotated. For instance, the \emph{CD4} term annotation (i.e., \textsf{lex}=CD4, \textsf{sem}=G\#protein\_molecule) is nested inside the \emph{CD4 promoter} annotation (i.e., \textsf{lex}=CD4\_promoter, \textsf{sem}= G\#DNA\_domain\_or\_region). We dropped two annotations found to lack lexical unit entries. There is a special Coordinated semantic class used for annotating biological terms across discontiguous regions of text. We dropped Coordinated class marked-up text from the analysis as it is troublesome to process and makes up only 1.6\% of all annotations.

Preprocessing the above sentence according to our procedure yields:
\begin{lstlisting}[belowskip=-1.0 \baselineskip]
in this report we show that the minimal |\colorbox{waterspout}{CD4\_promoter\_lex}| |\colorbox{water}{CD4\_lex}| has four |\colorbox{waterspout}{factor\_binding\_site\_lex}| each of which is required for full function
\end{lstlisting}
Annotated raw text is replaced by lexical units postfixed with ``\_lex''. The postfixing step is necessary to disambiguate annotated from unannotated text. For example, the lexical unit protease would be otherwise indistinguishable from unannotated usages of the term \emph{protease}. We included all lexical units recursively extracted nested \textsf{cons} tags. For unannotated text, we removed any non-ASCII terms, then converted terms to lowercase, then removed any punctuation marks, then substituted textual representations for any terms representing integers. We subjected all corpus documents to this procedure.

\subsection{Terminology extraction task}
To evaluate our proposed RICF measure, we extracted gold standard molecular biology terms from the preprocessed GENIA Term corpus dataset.

We conducted six separate experiments. In the main experiment, we take as ground truth lexical units from all five of the high-level semantic classes \textsf{Amino Acid}, \textsf{Nucleotide}, \textsf{Multi-cell}, \textsf{Cell}, and \textsf{Other}. The remaining five experiments are more fine-grained, each restricting ground truth terms to lexical units from a single high-level semantic class.

For each experiment, we compare RICF performance, as judged by the \textit{Precision at k} (P@k) metric, with the corresponding performances of five baseline methods: Chi-sq, CG, ICB, DoP, and KeyBERT. Chi-sq, CG, ICB, and DoP are the conventional term burstiness measures outlined in Section~\ref{subsec:bursty-measures}. KeyBERT, on the other hand, is a transformer-based document keyword extraction tool. We applied KeyBERT to our present setting by having it treat the entire collection as a single document. The KeyBERT N-gram range parameter we set to 1 (i.e., extract single terms as opposed to keyphrases) to put it on an equal footing with the term burstiness measures used in this study.

\begin{table*}
\centering
\begin{tabular}{lccccccrr}
\multicolumn{1}{c}{} & \multicolumn{5}{c}{\textbf{Baseline}} & \multicolumn{3}{c}{\textbf{Proposed}} \\
\cmidrule(rl){2-6} \cmidrule(rl){7-9}
\textbf{k} & {Chi-sq} & {CG} & {ICB} & {DoP} & {KeyBERT} & {RICF} & {\begin{tabular}{@{}c@{}}Improvement \\ over Chi-sq\end{tabular}} & {\begin{tabular}{@{}c@{}}Improvement \\ overall\end{tabular}} \\
\midrule
\multicolumn{9}{c}{All semantic classes} \\
\midrule
10 & \underline{\bf{1.0000}} & \underline{\bf{1.0000}} & \underline{\bf{1.0000}} & 0.0000 & \underline{\bf{1.0000}} & \bf{1.0000} & 0.00\% & 0.00\% \\
50 & 0.9400 & 0.9600 & \underline{\bf{1.0000}} & 0.0200 & \underline{\bf{1.0000}} & \bf{1.0000} & 6.38\% & 0.00\% \\
100 & 0.9400 & 0.9800 & 0.9800 & 0.0200 & \underline{\bf{1.0000}} & \bf{1.0000} & 6.38\% & -1.00\% \\
500 & 0.9680 & 0.9820 & 0.9740 & 0.0880 & \underline{\bf{1.0000}} & 0.9900 & 2.27\% & -1.60\% \\
1000 & 0.9590 & 0.9810 & 0.9640 & 0.1420 & \underline{\bf{1.0000}} & 0.9840 & 2.61\% & -1.30\% \\
5000 & 0.9146 & 0.9278 & 0.8972 & 0.3754 & \underline{\bf{1.0000}} & 0.9320 & 1.90\% & -6.80\% \\
\midrule
\multicolumn{9}{c}{\colorbox{water}{\textsf{Amino Acid} semantic class}} \\
\midrule
10 & 0.6000 & \underline{\bf{1.0000}} & 0.8000 & 0.0000 & 0.3000 & \bf{1.0000} & 66.67\% & 0.00\% \\
50 & 0.6000 & \underline{0.7600} & 0.7000 & 0.0200 & 0.1600 & \bf{0.8000} & 33.33\% & 5.26\% \\
100 & 0.5800 & \underline{0.8200} & 0.7400 & 0.0100 & 0.2000 & \bf{0.8300} & 43.10\% & 1.22\% \\
500  & 0.5500 & \underline{\bf{0.6900}} & 0.6220 & 0.0480 & 0.2340 & \bf{0.6900} & 25.45\% & 0.00\% \\
1000 & 0.5300 & \underline{\bf{0.6450}} & 0.5900 & 0.0700 & 0.2770 & 0.6420 & 21.13\% & -0.47\% \\
5000 & \underline{\bf{0.4296}} & 0.4288 & 0.4106 & 0.0163 & 0.2714 & 0.4292 & -0.09\% & -0.09\% \\
\midrule
\multicolumn{9}{c}{\colorbox{waterspout}{\textsf{Nucleotide} semantic class}} \\
\midrule
10 & \underline{\bf{0.3000}} & 0.0000 & 0.1000 & 0.0000 & 0.1000 & 0.0000 & -100.00\% & -100.00\% \\
50 & 0.1400 & 0.1000 & 0.1200 & 0.0000 & \underline{\bf{0.3000}} & 0.1000 & -28.57\% & -66.67\% \\
100 & 0.1500 & 0.1000 & 0.1000 & 0.0000 & \underline{\bf{0.3100}} & 0.1100 & -26.67\% & -64.52\% \\
500 & 0.1820 & 0.1420 & 0.1460 & 0.0040 & \underline{\bf{0.2840}} & 0.1400 & -23.08\% & -50.70\% \\
1000 & 0.1720 & 0.1390 & 0.1400 & 0.0140 & \underline{\bf{0.2240}} & 0.1370 & -20.35\% & -38.84\% \\
5000 & 0.1534 & \underline{0.1550} & 0.1482 & 0.0502 & 0.1494 & \bf{0.1558} & 1.56\% & 0.52\% \\
\midrule
\multicolumn{9}{c}{\colorbox{canary}{\textsf{Multi-cell} semantic class}} \\
\midrule
10 & \underline{\bf{0.0000}} & \underline{\bf{0.0000}} & \underline{\bf{0.0000}} & \underline{\bf{0.0000}} & \underline{\bf{0.0000}} & \bf{0.0000} & NA & NA \\
50 & 0.0000 & \underline{\bf{0.0400}} & \underline{\bf{0.0400}} & 0.0000 & 0.0000 & \bf{0.0400} & Inf & 0.00\% \\
100 & 0.0000 & \underline{\bf{0.0200}} & \underline{\bf{0.0200}} & 0.0000 & 0.0000 & \bf{0.0200} & Inf & 0.00\% \\
500 & 0.0300 & \underline{0.0340} & 0.0320 & 0.0060 & 0.0000 & \bf{0.0380} & 26.67\% & 11.76\% \\
1000 & 0.0320 & 0.0370 & \underline{\bf{0.0410}} & 0.0011 & 0.0000 & 0.0350 & 9.38\% & -14.63\% \\
5000 & \underline{0.0432} & 0.0426 & 0.0414 & 0.0202 & 0.0072 & \bf{0.0442} & 2.31\% & 2.31\% \\
\midrule
\multicolumn{9}{c}{\colorbox{menthol}{\textsf{Cell} semantic class}} \\
\midrule
10 & \underline{\bf{0.1000}} & 0.0000 & 0.0000 & 0.0000 & 0.0000 & 0.0000 & -100.00\% & -100.00\% \\
50 & \underline{\bf{0.0600}} & 0.0000 & 0.0000 & 0.0000 & 0.0000 & 0.0000 & -100.00\% & -100.00\% \\
100 & \underline{\bf{0.0600}} & 0.0100 & 0.0200 & 0.0100 & 0.0000 & 0.0100 & -83.33\% & -83.33\% \\
500 & \underline{\bf{0.0780}} & 0.0360 & 0.0600 & 0.0200 & 0.0440 & 0.0400 & -48.72\% & -48.72\% \\
1000 & \underline{\bf{0.0820}} & 0.0580 & 0.0690 & 0.0260 & 0.0600 & 0.0620 & -24.39\% & -24.39\% \\
5000 & \underline{0.0958} & 0.0992 & 0.0988 & 0.0514 & 0.1540 & \bf{0.1002} & 4.59\% & 4.59\% \\
\midrule
\multicolumn{9}{c}{\colorbox{light-red}{\textsf{Other} semantic class}} \\
\midrule
10   & 0.0000 & 0.0000 & 0.1000 & 0.0000 & \underline{\bf{0.6000}} & 0.0000 & -100.00\% & -100.00\% \\
50   & 0.1400 & 0.0600 & 0.1400 & 0.0000 & \underline{\bf{0.5400}} & 0.0600 & -57.14\% & -88.89\% \\
100  & 0.1500 & 0.0300 & 0.1000 & 0.0000 & \underline{\bf{0.4900}} & 0.0300 & -80.00\% & -93.88\% \\
500  & 0.1460 & 0.0800 & 0.1140 & 0.0100 & \underline{\bf{0.4380}} & 0.0820 & -43.84\% & -81.28\% \\
1000 & 0.1370 & 0.1020 & 0.1240 & 0.0210 & \underline{\bf{0.4390}} & 0.1080 & -21.17\% & -75.40\% \\
5000 & 0.1926 & 0.2022 & 0.1982 & 0.0906 & \underline{\bf{0.4180}} & 0.2026 & 5.19\% & -51.53\% \\
\midrule
\end{tabular}
\caption{GENIA Term corpus P@k scores accompanied by percent improvement of RICF over the chi-squared test and the top scoring baseline method.}
\label{tbl:genia-p-at-k}
\end{table*}

Table~\ref{tbl:genia-p-at-k} shows the results of our molecular biology terminology extraction task experiments. Each block in the table corresponds to an experiment. In each block, RICF and baseline method P@k scores are recorded for $k = 10$, $50$, $100$, $500$, $1000$, $5000$. Bolded results represent the best performance in an experimental setting for a given value of $k$. Underlined results represent the best performance among the baseline methods. The second last column of the table records the percent improvement over/under the chi-square test. The last column of the table records the percent improvement of RICF over/under the best scoring scoring baseline method.

Consider first the main experiment which is summarized in the uppermost block of Table~\ref{tbl:genia-p-at-k}. It is apparent that KeyBERT exhibits the best overall performance. This is to be expected as underlying KeyBERT is a complex algorithm that capitalizes on pre-trained word embeddings. Note, however, that RICF is nevertheless competitive with each of KeyBERT and RICF registering perfect P@10, P@50, and P@100 scores. On the other hand, RICF generally outperforms the classical term burstiness scores, Chi-sq in particular. The DoP score shows particularly poor performance. 

Let us now turn to results as broken down by high-level semantic class as shown in the color-coded blocks of Table~\ref{tbl:genia-p-at-k}. The results show RICF to be competitive with the baseline measures, including the state-of-the-art KeyBERT keyword extraction tool. While no single measure ranks as universally best, some general observations may be drawn. First, the DoP measure perform poorly across all experiments. Second, RICF generally outperforms CG and ICB head-to-head match ups. Interestingly, RICF performance tracks closely to that of CG, yet RICF performs at least as well as CG in~26 out of~30 evaluation settings. A similar, although less pronounced, pattern is observed between RICF and ICB with RICF exhibiting better or equal to that of ICB in 21 out of 30 evaluation settings. Comparing RICF and Chi-sq, we find that RICF exhibits superior performance to Chi-sq in~10 out of~30 experimental settings, lower performance in~16, and ties it in~4. The head-to-head match up between RICF and KeyBERT demands a more nuanced interpretation. As we saw, KeyBERT edges out RICF when biological terms from all semantic classes are taken as ground truth. However, the high number of annotations in the GENIA data leads to nearly perfect P@k scores, making comparison difficult. KeyBERT, however, exhibits clear superior performance in the \textsf{Nucleotide} and \textsf{Other} semantic class experiments. By contrast, RICF decidedly outperforms KeyBERT in the \textsf{Amino Acid} and \textsf{Multi-cell} semantic class experiments. In the overall, RICF performs at least as well as KeyBERT in~18 out of~30 experimental settings.

In summary, we find that RICF compares favorably to classical term burstiness measures, as well as the transformer-based method KeyBERT, on terminology extraction when all terminology classes are considered. The DoP measure exhibits especially poor performance. The remaining term burstiness measures Chi-sq, CG, ICB, and RICF show generally comparable performance with RICF registering as the best of the four as measured by the number of top place finishes (i.e.,~3 out of~6). The results of the semantic-class specific experiments are more difficult to interpret. On the whole, however, these results demonstrate the promise of RICF as a method for quantifying term burstiness at the corpus level.

\subsection{Stopwords exploratory analysis}
Stopwords are common words traditionally considered to not carry substantial information on their own (e.g., \emph{the}, \emph{as}, \emph{but}, etc.). It stands to reason that a term burstiness measure worth its salt, and terminology extraction method more generally, will rank stopwords very lowly. To assess the term burstiness measures and KeyBERT by this standard, we compiled a list of~989 English stopwords by pooling stopwords from the \textsf{nltk} 3.8.1 Python library (179~stopwords), the Terrier IR Platform (733~stopwords) (\cite{Macdonald2012}), and MyISAM (543~stopwords) (\cite{MyISAM2023}). In subsequent numerical investigations, we considered only the~417 out of~989 stopwords found to occur in the preprocessed GENIA data. The results shed some light on our findings in the keyword extraction task experiment.

\begin{table*}[!ht]
\setlength{\tabcolsep}{1.5mm}
\centering
\begin{tabular}{lrrrrrr}
\textbf{Rank} & {Chi-sq} & {CG} & {ICB} & {DoP} & {KeyBERT} & {RICF} \\
\midrule
1  &    flavonoid &       Bcl-6 &         Bcl-6 &   of & human\_monc &       Bcl-6 \\
2  &   immunorece &      v-erbA &       TCRzeta &  the & cell-type- &       v-erbA \\
3  &        TCF-1 &         SMX &          ML-9 &   in & T\_cell\_IL3 &        SMX \\
4  &  aNF-E2\_pro &        ML-9 &          AITL &  and & T-cell\_grow &       SHP1 \\
5  &          p40 &        SHP1 &          SHP1 &   to & cytosolic\_ &        ML-9 \\
6  &         LPMC & beta-casein &   beta-casein &    a & mouse\_inte & beta-casein \\
7  &           GP &      EBNA-2 &         A-myb & that & interleuki &       p95vav \\
8  &  c-sis/PDGF- &   I\_kappaB &     I\_kappaB &   by & T\_cell\_int &         DM \\
9  & Ad\_2/5\_E1a &          DM &           SMX & with & IL-4-activ &      TCRzeta \\
10 &  NF-IL6\_gen &     TCRzeta & Rap1\_protein &   we & human\_inte &   I\_kappaB \\
\midrule
\textbf{Pct.} & 0\% & 0\% & 0\% & 100\% & 0\% & 0\% \\
\midrule
\end{tabular}
\caption{Top 10 ranking terms by term burstiness measure. Extracted terms are truncated at 10 characters. The bottom row is the percent of top 10 terms that are stopwords.}
\label{tbl:stopwords-top-10}
\end{table*}

The top~10 ranking terms for each method are shown in Table~\ref{tbl:stopwords-top-10}. DoP is made conspicuous by virtue that its top~10 terms are all stopwords, as opposed to the other methods which count no stopwords among their top~10 ranked terms. Another salient feature is that CG, ICB, and RICF share~8 out of~10 of their top~10 terms in common. This hints that RICF is a incremental improvement over CG and ICB. Interestingly, Chi-sq shares no top~10 terms in common with any of CG, ICB, and RICF. This reinforces our previous finding that Chi-sq behaves fundamentally differently from the three apparently related measures CG, ICB, and RICF. Although less surprising given it is based on a fundamentally different methodology, KeyBERT's top~10 terms are likewise completely different to those identified by the other methods. 

\begin{table*}[!ht]
\centering
\begin{tabular}{lrrrrr}
\textbf{Method} & Min & Q1 & Median & Q3 & Max \\
\midrule
Chi-sq & 94 & 8,023 & 8,798 & 39,152 & 40,803 \\
CG & 19 & 7,751 & 8,755 & 23,954 & 40,681 \\
ICB & 70 & 11,311 & 16,853 & 22,344 & 40,505 \\
DoP & 1 & 254 & 1,112 & 3,719 & 40,488 \\
KeyBERT & 26,477 & 34,116 & 35,872 & 38,210 & 40,801 \\
RICF & 2,094 & 8,123 & 8,811 & 39,419 & 40,803 \\
\midrule
\end{tabular}
\caption{Stopwords ranking five-number summary.}
\label{tbl:stopwords-five-point-summary}
\end{table*}

To get a statistical sense of how the methods handle stopwords, we generated a standard five-number summary of stopword rankings for each method. The resulting statistics are summarized in Table~\ref{tbl:stopwords-five-point-summary}. For each method (first column), we report the topmost rank (second column), the first to third quartile ranks (third to fifth columns), and the maximum rank (sixth column) among the~417 stopwords. It is here when RICF truly distinguishes itself from the other term burstiness measures. Consider that RICF's highest ranking stopword occurs at the 2,094'th place in the term ranking. By contrast, we find the same figure is~94 for Chi-sq,~19 for CG,~70 for ICB, and~1 for DoP. Granted the differences in rankings are less dramatic over the rest of the spread. That said, this outcome points to RICF doing something right as a term burstiness measure. KeyBERT is impressive for the degree to which it suppresses stopwords in its term rankings. As shown in the table, a stopword does not appear in the KeyBERT term rankings until the 26,477'th position. This may explain how KeyBERT was able to achieve perfect performance in the ``All'' semantic classes experiments, while the term burstiness measures fell short.

\section{Discussion} \label{sec:discussion}
The work we present has a number of practical and theoretical implications.

The RICF term burstiness measure introduced in this study highlights the as yet unrealized potential of significance testing in text analysis. Unlike the traditional chi-squared test, which relies on ICF, RICF leverages both ICF and IDF in quantifying term burstiness. By leveraging both ICF and IDF, the RICF measure outperformed the chi-squared test at technical terminology extraction and stopword filtering. RICF, moreover, generally outperformed other traditional measures of term burstiness (i.e., CG, ICB, and DoP) at the same tasks. On the other hand, transformer-based document keyword extraction tool KeyBERT showed the best overall performance at technical term extraction and stopword filtration. An exciting new area of research involves enhancing word embeddings with document-class context by fusing them with term weightings trained on class labels~(\cite{Kim2019,Sun2022,Xie2022}. Fusing KeyBERT word embeddings with RICF-based term weightings in this manner presents a future application with the potential to improve the state-of-the-art in terminology extraction.

RICF serves as a proxy for our proposed exact test. A limitation of RICF lies in its exploiting only an expected value to model associated exact test conditional distribution tail probabilities. In the future, approximating the exact test tail probability by means of Poissonizing the multinomial distribution underlying our language modeling approach should be explored. An additional challenge will be to incorporate smoothing techniques into the estimation process to better model rare term occurrences~(\cite{Zhai2017}).

In a recent review paper, \cite{Nomoto2022} identifies four varieties of keyword: 1) terminology, 2) topics, 3) index terms, and 4) summary terms. We chose as our primary application terminology extraction because RIFC is specially suited to extract that variety of keyword. To understand why, consider ``the whelks problem'' as first articulated by~\cite{Kilgarriff1997}. \cite{Sharoff2017} puts it concisely: ``If you have a text about whelks, no matter how infrequent this word is in the rest of your corpus, it’s likely to be in nearly every sentence in this text.'' The \emph{whelk}-like terms in the GENIA corpus, which are technical terms in Nomoto's nomenclature, will tend to lie at the far right-hand-side of Fig.~\ref{fig:genia-idf-vs-icf}. Moreover, that many of the largest RICF scores associate to these \emph{whelk}-like terms is apparent from visual inspection. With this observation in mind, it should be possible to develop novel significance tests, and associated RICF-variants, that better model non-technical terms. As a first step, it would be interesting to survey the landscape of keyword types (i.e., terminology, topics, index terms, and summary terms) as the lie in the plot of Fig.~\ref{fig:genia-idf-vs-icf}. For example, if index terms tend to concentrate in the middle of the plot, then novel significance tests could be derived to place maximum emphasis on that region. Furthermore, integrating RICF, as well as future RICF-variants tailored to model other varieties of keywords, as a global component in term weighting models opens the door to a host of future applications in document classification, document summarization and document retrieval~(\cite{Rathi2023}). 

Evaluation of the RIFC heuristic relies on a relationship (in expectation) relating IDF to ICF under the multinomial language model. The deriving of this relationship is a tangible theoretical contribution of our work, and a starting point for better understanding such relationships under more complex language models. A logical next step concerns generalizing the result to the Dirichlet-multinomial language model~(\cite{Cummins2015,Cummins2017,You2022}).

\section{Conclusion} \label{sec:conclusion}
This paper advances an exact test of statistical significance for the identification of information rich terms, known as bursty terms, in document collections. The exact test, which is computationally intractable, inspired our proposal of the RICF term burstiness heuristic measure. In the process, we derived an equation relating IDF to ICF in expectation under the multinomial language model. RICF is verified by domain-specific terminology extraction experiments on the GENIA Term corpus benchmark. The results show RICF compares favorably to existing term burstiness measures. It is our hope that this work motivates future investigations into the use of statistical significance testing in text analysis.

\section*{Acknowledgements}
The work was funded by University of Prince Edward Island Start-up Grant M605391.

\section*{Data and code availability}
The data and code used to analyze it are available at the GitHub repository \url{https://github.com/paul-sheridan/bursty-term-measure}, Commit Id: fde551ff41dc77a6a503926ff78016f8a6130932.

\printcredits

\appendix
\section{Technical details} \label{appendix:proofs}

\subsection{ICF expected value} \label{appendix:eicf}
In this section, we derive the expected value of ICF under the multinomial language model. Let the random variable $N_i$ denote the number of times the $i$'th term occurs in a collection,~$\mathcal{C}_{d,m}$.

\begin{theorem} \label{thrm:eicf}
The expected value of the ICF for the $i$'th term in a collection, $\mathcal{C}_{d,m}$, is given by
\begin{equation} \label{eqn:eicf}
    \mathbf{E}[icf(i)] = \frac{1-\theta_i}{2n\theta_i} - \log\theta_i + o(n^{-2}).
\end{equation}
\end{theorem}

\begin{proof} A little algebra followed up by an application of the second order Taylor series expansion of $\log(N_i)$ centered on $\mathbf{E}[N_i]$ yields
\begin{align*}
  \mathbf{E}[icf(i)] &= \mathbf{E}[-\log\left(N_i /n\right)], \\
  &= -\mathbf{E}[\log(N_i)] + \log(n), \\
  &\approx -(\log\mathbf{E}[N_i] - \mathbf{V}[N_i]/(2\mathbf{E}[N_i]^2) + \log(n).
\end{align*}
This particularizes to
\begin{align*}
  \mathbf{E}[icf(i)] &\approx -\log(n\theta_i) + \frac{n\theta_i (1-\theta_i)}{2(n\theta_i)^2} + \log(n) \\ 
    &= \frac{1-\theta_i}{2n\theta_i} - \log\theta_i
\end{align*}
since $N_i$ is binomially distributed with mean $\mathbf{E}[N_i] = n\theta_i$ and variance $\mathbf{V}[N_i]=n\theta_i (1-\theta_i)$.

It remains to show that the higher order terms are in little-O of $n^{-2}$. Consider the Taylor series expansion
\begin{equation*}
\mathbf{E}[\log(N_i)] = \log(n\theta_i) - \frac{\mathbf{V}(N_i)}{2 (n \theta_i)^2} + \sum_{k=3}^\infty (-1)^{k+1} \frac{\mathbf{E}[N_i - n\theta_i]^k}{k(n\theta_i)^k}.
\end{equation*}

It will suffice to show that when $n \theta_i$ is large,
\begin{equation*}
\mathbf{E}[\log(N_i)] \sim \log(n \theta_i) - \frac{\mathbf{V}(N_i)}{2(n \theta_i)^2}.
\end{equation*}

We want to show that all $\frac{(N_i - n \theta_i)^k}{k (n \theta_i)^k}$ terms from the Taylor series expansion with $k \geq 3$ are trivially small as $n$ tends to infinity. In order to show this, we will need to bound $N_i$. We will be interested in large $n$ values (approaching infinity). Note that $\theta_i$ is constant, so the mean, $n \theta_i$, will also approach infinity. Since $N_i$ is binomial, we can use the Chernoff bound of~\cite{Mulzer2018}:
\begin{equation*}
Pr(N_i \geq k) \leq exp\left(-n\left(\frac{k}{n}\right)\log\left(\frac{k}{n \theta_i}\right) + \left(1 - \frac{k}{n}\right) \log\left(\frac{1 - \frac{k}{n}}{1 - \theta_i}\right)\right).
\end{equation*}

Letting $k =  n \theta_i + t$ for $t>0$, we have
\begin{align*}
Pr(X \geq n \theta_i + t) &\leq exp\left[-n \theta_i\log\left(1 + \frac{t}{n \theta_i}\right) - t\log\left(1 + \frac{t}{n \theta_i}\right) + \log\left(1 - \frac{t}{n(1 - \theta_i)}\right)\right.\\\\
& \left. - \theta_i\log\left(1 - \frac{t}{n(1 - \theta_i)}\right) - \frac{t}{n}\log\left(1 - \frac{t}{n(1 - \theta_i)}\right)\right].\\\\
\end{align*}
We momentarily ignore the base and focus on the limit of the exponent as $n$ approaches infinity. Many of the terms approach zero as $n$ approaches infinity. Applying L'H\^{o}pital's rule to the remaining terms gives: 
\begin{equation*}
\lim\limits_{n \to \infty} - n \theta_i \times \log\left(1 + \frac{t}{ n \theta_i}\right) = \lim\limits_{n \to \infty} -\frac{\log\left(1 + \frac{t}{n \theta_i}\right)}{\frac{1}{ n \theta_i}} = -t
\end{equation*} 
Recombining this with the base, $e$, we get that the limit is $e^{-t}$. So, the probability that $N_i$ is greater than $ n \theta_i + t$ is less than $e^{-t}$. For large $t$, this is a small probability. Below, we evaluate the magnitude of a given term in the Taylor series expansion if $N_i =  n \theta_i + t$. Since we can say with high probability that this result is greater than any real value that $N_i$ will take on, we can make an argument regarding the growth rate of the terms in the Taylor series expansion. 

\begin{align*}
\mathbf{E}\left[\frac{(N_i -  n \theta_i)^k}{ (n \theta_i)^k}\right] &= \frac{\mathbf{E}[(( n \theta_i + t)^k - \binom{k}{1}(n\theta_i + t)^{k-1} n\theta_i + \binom{k}{2}(n \theta_i + t)^{k-2}(n\theta_i)^2 - \dots)]}{(n\theta_i)^k},\\\\
       &= \frac{t^k}{(n\theta_i)^k}.
\end{align*}

As $n$ tends to infinity, this upper bound shows that the terms with degree 3 and greater from the Taylor series expansion tend to zero and, since each individual term is bounded by a constant factor of $(n \theta_i)^{-k}$, the entire sum from the degree 3 term onward is o($n^{-2}$), as we set out to demonstrate.
\end{proof}

\subsection{IDF expected value} \label{appendix:eidf}
The expected value of IDF is obtained in a similar manner to that of ICF. Let the random variable $B_i$ denote the number of documents in a collection, $\mathcal{C}_{d,m}$, containing the $i$'th term.

\begin{theorem} 
The expected value of the IDF for the $i$'th term in a collection, $\mathcal{C}_{d,m}$, is given by
\begin{equation} \label{eqn:eidf}
  \mathbf{E}[idf(i)] = \frac{\sum_{j=1}^d \rho_i^{n_j} - \sum_{j=1}^d \rho_i^{2n_j}}{2d^2 \left(1-\frac{1}{d} \sum_{j=1}^d \rho_i^{n_j} \right)^2} - \log\left(1-\frac{1}{d} \sum_{j=1}^d \rho_i^{n_j}\right) + o(d^{-2}),
\end{equation}
where $\rho_i = 1 - \theta_i$.
\end{theorem}
    
\begin{proof} The same second order Taylor series expansion approach employed in the proof of Theorem~\ref{thrm:eicf} gives
\begin{align*}
  \mathbf{E}[idf(i)] &= \mathbf{E}[-\log\left(B_i /d\right)] \\
  &= -\mathbf{E}[\log(B_i)] + \log(d) \\
  &\approx -\log\mathbf{E}[B_i] + \mathbf{V}[B_i]/(2\mathbf{E}[B_i]^2) + \log(d)
\end{align*}
Unlike $N_i$, which is binomial with parameters $\theta_i$ and $n_j$, the random variable $B_i$ is Poisson binomially distributed according to
\begin{equation*}
  Pr(B_i = b_i) = \sum_{A \in F(b_i)} \prod_{r \in A}\phi_{ir} \prod_{s \in A^c}(1-\phi_{is}),
\end{equation*}
where $\phi_{ij} = 1 - (1-\theta_i)^{n_j}$, and $F(b_i)$ is the set of all subsets of $b_i$ integers that can be selected from $\left\{1,2, \ldots, d \right\}$. This follows from the facts that $Pr(N_{ij}=0)=(1-\theta_i)^{n_j} \equiv 1-\phi_{ij}$ and $Pr(N_{ij} \geq 1)=1-Pr(N_{ij}=0)=1-(1-\theta_i)^{n_j} \equiv \phi_{ij}$. Plugging in the mean $\mathbf{E}[B_i]=d-\sum_{j=1}^d(1-\theta_i)^{n_j}$ and variance $\mathbf{V}[B_i]=\sum_{j=1}^d(1-\theta_i)^{n_j} - \sum_{j=1}^d(1-\theta_i)^{2n_j}$ and simplifying gives
\begin{align*}
  \mathbf{E}[idf(i)]  &\approx -\log\left(d-\sum_{j=1}^d(1-\theta_i)^{n_j}\right) + \frac{\sum_{j=1}^d(1-\theta_i)^{n_j} + \sum_{j=1}^d(1-\theta_i)^{2n_j}}{2\left(d-\sum_{j=1}^d(1-\theta_i)^{n_j}) \right)^2} + \log(d), \\ 
  &= \frac{\sum_{j=1}^d \rho_i^{n_j} - \sum_{j=1}^d \rho_i^{2n_j}}{2d^2 \left(1-\frac{1}{d} \sum_{j=1}^d \rho_i^{n_j} \right)^2} - \log\left(1-\frac{1}{d} \sum_{j=1}^d \rho_i^{n_j}\right).
\end{align*}

We would like to bound $B_i$ as $d$ gets large. An upper bound on $B_i$ can be created by relating it to a binomial random variable with constant probability $p$ equal to the maximum of the~$\phi_{ij}$ values. Since this only increases the probability of given Bernoulli outcomes, it serves as an upper bound. Using the same argument as above with the binomial random variable~$N_i$, replacing the $n$'s and $\theta$'s in the above proof with $d$'s and $\phi_{ij}$'s, this bound will leave us with a negligible distance from the mean as $d$ approaches infinity. Similar to before, we see that the sum of the terms in the Taylor series from the cubic term onward is o($d^{-2}$). 
\end{proof}

\begin{corollary} Applying the approximations $(1-\theta_i)^{n_j} \approx e^{-n_j \theta_i}$ for $\theta_i$ close to $0$ and $e^{-n_j \theta_i} \approx e^{\mu\theta_i}$ gives
\begin{equation} \label{eqn:eidf-approx}
     \mathbf{E}[idf(i)] \approx \frac{e^{-\mu\theta_i}}{2d(1 - e^{-\mu\theta_i})} - \log(1 - e^{-\mu\theta_i}) + o(d^{-2}),
\end{equation}
where $\mu = \mathbf{E}[N_j]$ is the expected size of the $j$'th document.
\end{corollary}

\subsection{ICF conditional expectation evaluation} \label{appendix:cond-expectation-algorithm}
The following procedure was used to estimate the conditional expectation $\mathbf{E}[icf(i)|B_i = b_i]$ of Eq.~(\ref{eqn:ricf}). The strategy amounts to estimating the value of $\theta_i$, call it $\hat\theta_i$, satisfying Eq.~(\ref{eqn:idf-icf-expected-relation}) for the $i$'th term with observed TDF $b_i$. To this end, we used the Python library \textbf{scipy}~1.11.13 implementation of Brent's method to find the root of the function $\mathbf{E}[idf(i)] - idf(i)$ on the interval $[1/n, n_i^{*}/n]$, where $n_i^{*} = max(\left\{n_i : i=1, \ldots, m\right\})$. This root is our $\hat\theta_i$. Plugging in~$\hat\theta_i$ into Eq.~(\ref{eqn:eicf}) yields the desired estimate for $\mathbf{E}[icf(i)|B_i = b_i]$.

\bibliographystyle{cas-model2-names}

\bibliography{bibliography}

\end{document}